\documentclass[submission,copyright,creativecommons]{eptcs}
\providecommand{\event}{NCMA 2026} % Name of the event you are submitting to

\usepackage{iftex}

\ifpdf
  \usepackage{underscore}         % Only needed if you use pdflatex.
  \usepackage[T1]{fontenc}        % Recommended with pdflatex
\else
  \usepackage{breakurl}           % Not needed if you use pdflatex only.
\fi
\usepackage{amsthm}
\usepackage{amsmath,amsfonts,amssymb,mathtools}
\usepackage{stmaryrd}
\usepackage{tabularx}
\usepackage{array}
\usepackage{longtable}
\usepackage[bottom]{footmisc}
\usepackage{tikzsymbols}
\usepackage{tikz}
\usetikzlibrary{automata, positioning}
\usetikzlibrary{arrows}
\usetikzlibrary{arrows.meta,cd}
\tikzset{
  hookrightarrow'/.style = {Hooks[left]->}
}
\usetikzlibrary{calc}
\usetikzlibrary{patterns}
\usetikzlibrary{decorations.pathreplacing}
\usepackage{subcaption}
\newcommand{\startofpushdown}{\triangleleft}

\makeatletter
\providecommand{\theHexample}{\the\c@example}

\theoremstyle{plain}
\newtheorem{theorem}{Theorem}[section]

\newtheorem{proposition}{Proposition}[section]

\theoremstyle{definition}
\newtheorem{definition}{Definition}[section]
\newtheorem{example}{Example}[section]
\newtheorem{remark}{Remark}[section]
\newtheorem{illustration}{Illustration}[section]

\title{2-Head 2D Returning Finite Automata}
\author{
Henning Fernau
\institute%
{Abteilung Informatikwissenschaften,\\
  Fachbereich 4, Universit\"at Trier,\\
  54286 Trier, Germany}
  \email{fernau@uni-trier.de}
  \and
Benedek Nagy
\institute%
{Eastern Mediterranean University \\ Famagusta, North Cyprus, via Mersin-10, Türkiye \\ Eszterházy Károly Catholic University, Eger, Hungary}
  \email{nbenedek.inf@gmail.com}
\and
R. Jennifer Rose
\institute%
{Madras Christian College\\
  Chennai -- 600059, Tamil Nadu, India}  \email{jennifer.royson@gmail.com}  
  \and
Robinson Thamburaj
\institute%
{Madras Christian College\\
  Chennai -- 600059, Tamil Nadu, India}
  \email{robinson@mcc.edu.in}
  \and
D. Gnanaraj Thomas
\institute%
{Madras Christian College\\
  Chennai -- 600059, Tamil Nadu, India}
  \email{dgthomasmcc@yahoo.com}
  }
\def\titlerunning{2-Head 2D Returning Finite Automata}
\def\authorrunning{H. Fernau, B. Nagy, R.J. Rose, R. Thamburaj \& D.G. Thomas}
\def\copyrightholders{Fernau, Nagy, Rose, Thamburaj, Thomas}

\begin{document}
\maketitle

\renewcommand{\arraystretch}{1.3}

\begin{abstract}
We introduce and study a family of two‑head finite automata called  two‑head returning finite automata (2‑HRFA) operating on rectangular arrays of picture languages, in which both heads move in opposite directions. We show that the class of picture languages accepted by 2‑HRFA is incomparable with the class of languages generated by context‑free matrix grammars (CFMG), 
while it forms a proper subset of the class of languages accepted by returning pushdown automata (RPDA). 
 In addition, we define a constrained variant, both‑head stepping two‑head returning finite automata (B2-HRFA), in which both heads are required to move in a synchronized, stepwise fashion. We prove that the class of languages accepted by  returning finite automata (RFA) is a proper subset of the class of languages accepted by B2-HRFA, which in turn is a proper subset of the class of languages accepted by 2‑HRFA. Closure properties for both the families of languages accepted by 2‑HRFA and B2-HRFA are also investigated.
\end{abstract}

\section{Introduction}

The study of multi-head automata has played an important role in the theory of computation for understanding the computing capabilities of finite state devices. The multi-head automata introduced by Rosenberg \cite{Rosenberg1966} extend the classical one-head model by equipping the automaton with several reading heads where each reading head can move only left to right and is independent of each other, resulting in a strict hierarchy of language classes.

Two-head automata represent an important subclass of multi-head machines. Hromkovic studied two-head deterministic finite state automata~\cite{Hromkovic1985}, establishing the fundamental results on their computational complexity and demonstrating that one-way two-head automata can recognize languages that are not recognized by single-head machines.
2-head models were studied by Loukonova \cite{Loukanova2007} and by Nagy \cite{
NagyiConcept2010} for accepting the linear context-free languages. Deterministic variants accept a language class referred as 2detLIN, which is properly between the classes REG
of regular and LIN of linear languages~\cite{NagyParchami2021}.
Here we consider these automata for processing 2-dimensional rectangular inputs.  

The theory of string recognizability by automata is generalized to two dimensions in \cite{BlumHewitt1967,GiaRes97}. For recognizing rectangular-shaped pictures (or arrays), returning finite automata (RFA) were introduced by Fernau \emph{et al.} \cite{FerPST2018} as simpler finite-state devices, scanning rectangular pictures row by row, each row always from left to right. It was shown that the class of languages recognized by RFA is precisely the transpose of the class of languages generated by regular matrix grammars (RMG) introduced by Siromoney \emph{et al.} much earlier~\cite{SirSirKri72}. Recently, Fernau \emph{et al.} \cite{FerJen2025} introduced and studied   returning pushdown automata (RPDA) which extend the model of returning finite automata by equipping them with a pushdown store. 
This combination of returning behavior and pushdown store increases the expressive power beyond that of RFA and also beyond transposes of context-free matrix languages as defined in~\cite{SirSirKri72}.

In this paper, we introduce  two‑head returning finite automata (2‑HRFA) for recognizing  picture languages, in which the two heads scan a rectangular picture from opposite ends.\enlargethispage{4\baselineskip}\footnote{These automata are essentially different from the model introduced and studied in \cite{FerFreHol98c,FerFreHol99}, as the notion of an \emph{array} differs.} Specifically, the first head reads the picture from the top, moving left to right on each row, while the second head reads the picture from the bottom, moving right to left on each row and the automaton accepts the picture when the two heads meet. We compare the language‑theoretic properties of 
 the classes of 2‑HRFA with returning finite automata (RFA), regular matrix grammars (RMG), context-free matrix grammars (CFMG), returning pushdown automata (RPDA), and the constrained variant of both‑head stepping two‑head returning finite automata (B2-HRFA). We show that 2‑HRFA recognizes a strictly larger class of languages than RFA and that the class of languages accepted by 2‑HRFA is incomparable with the classes generated by RMG and CFMG. Furthermore, we establish that the class of languages recognized by 2‑HRFA lies strictly within the class of languages accepted by RPDA, while the class of languages accepted by B2-HRFA forms an intermediate level between the classes of languages accepted by RFA and 2‑HRFA, thereby refining the hierarchy among the families of returning finite‑state and stack‑based models.

\section{Preliminaries}
In this section, we recall some definitions related to picture languages, matrix grammars and two-head automata, returning finite automata and returning pushdown automata \cite{FerJen2025,FerPST2018,GiaRes97,GinSpa66a,KriSir74a,Nagy2012}.

Let $\Sigma$ be a finite alphabet. A string or a word over $\Sigma$ is a finite sequence of symbols from $\Sigma$. The set of all non-empty strings over $\Sigma$ is denoted by $\Sigma^+$ and the empty string is $\varepsilon$; $\Sigma^* = \Sigma^+ \cup \{\varepsilon\}$. 
Further, let $\mathbb N$ denote the set of positive integers.

\begin{definition}
   A 2-Head Finite Automaton (2-HFA) in which the two heads read a word in opposite directions is defined as $A=(Q,\Sigma,\delta,s,F)$ where $Q$ is the finite set of states, $\Sigma$ is the input alphabet, $\delta:Q\times (\Sigma\cup \{\varepsilon\})\times (\Sigma\cup \{\varepsilon\})\rightarrow 2^Q$ is the transition function, $s\in Q$ is the start state, $F\subseteq Q$ specifies the final states.   
\end{definition}   

Suppose $a_1a_2\cdots a_n$ is an input word. The 2-HFA starts the computation in the initial state $s$ reading either a pair of symbols $(a_1,a_n)$ where the first head reads the symbol $a_1$ from the left end and the second head reads the symbol $a_n$ from right end and enters the next state according to the transition function $\delta$ or reading the pair $(a_1,\varepsilon)$ when only the first head moves or the pair $(\varepsilon,a_n)$ when only the second head moves. (We do not allow transitions by the pair $(\varepsilon,\varepsilon)$ as such transitions can be eliminated without affecting the class of accepted languages.)
The 2-HFA continues to read the remaining symbols of the input in a similar manner. An input word is accepted by a 2-HFA if every letter of the word is read by one of the heads (i.e., the heads meet) and the automaton is in a final state. In transitions, a pair of symbols $(a,b)\in\Sigma^2$ is assigned to the arrows between two states which means that when moving from one state to  another one, the first head reads the symbol $a$, the second head reads $b$ and both heads step forward. 

\begin{example}
    The language $L=\{a^nb^n\mid n>1\}$ is
    recognized by the 2-HFA given in \figurename~\ref{fig:Lsubatonbton}. Here, the two heads meet at $s_2$ to accept a word of the form $a^nb^n$.
\end{example}
\begin{figure}[t]
		\centering
		% States
				\begin{tikzpicture}[shorten >=1pt, node distance=7em, on grid, auto,thick,initial text=]

% Define states
					\node[state, initial] (s) {$s$};
					\node[state, right=of s] (s1) {$s_1$};
					\node[state, right=of s1,accepting] (s2) {$s_2$};
% Transitions
					\draw[->] (s) edge node[align=center] {$(a,b)$} (s1);
                    \draw[->] (s1) edge node[align=center] {$(a,b)$} (s2);
                    \draw[->] (s2) edge [loop above]node[align=center] {$(a,b)$} ();
                     \end{tikzpicture}
		\caption{A 2-head finite automaton that accepts the language $a^nb^n$ with $n>1$. \label{fig:Lsubatonbton}}
	\end{figure}
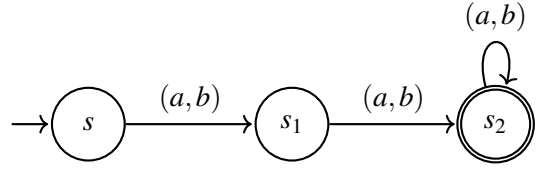

It is well-known that linear languages can be characterized by linear context-free grammars as well as by pushdown automata whose pushdown store makes only one turn, i.e., it only once switches between a pushing and a popping phase. 

    \begin{theorem}\label{thm:2HFA}
        A language is linear if and only if it is accepted by a 2-head finite automaton if and only if it is accepted by a  one-turn pushdown automaton.
    \end{theorem}

    The deterministic version of a 2-HFA is weaker than the non-deterministic version as can be seen by the linear language $\{a^nb^n\mid n\geq 1\}\cup \{a^{3n}b^n\mid n\geq 1\} $ which cannot be accepted by a deterministic 2-HFA~\cite{NagyParchami2021,NagyParchamiSadeghiAFL17}.
    The following example not only illustrates determinism but also the role of transitions where only one of the heads is moving, e.g., by  $(a,\varepsilon)$. 
  
\begin{example} \label{ex:Lsub{w=wtor}}
    The language $\{ww^R\mid w\in \Sigma^* \text{ and $w^R$ is the reversal of the word $w$}\}$ is accepted by a 2-HFA having a single state with the transition $(a,a)$ for each letter $a\in\Sigma$. The language of palindromes is accepted by another 2-HFA by adding a second accepting state to the previous 2-HFA and joining the two states with the transition $(a,\varepsilon)$ (for each $a\in\Sigma$). If the automaton should also work deterministically, we need, apart from the start state $q_0$, $|\Sigma|$ many states $q_a$ that are entered from $q_0$ when reading $(a,\varepsilon)$ and lead from $q_a$ into $q_0$ again when reading $(\varepsilon,a)$. All states are then accepting.
\end{example}
   
\begin{figure}[b]
    \centering
    \begin{minipage}{.48\textwidth}
        $$P \obar Q =  \begin{array}{ccc} a_{1,1} & \cdots & a_{1,n} \\ 
		\vdots & \ddots & \vdots \\ 
		a_{m,1} & \cdots & a_{m,n}
	\end{array} 
	\begin{array}{ccc} b_{1,1} & \cdots & b_{1,n'} \\ 
		\vdots & \ddots & \vdots \\ 
		b_{m',1} & \cdots & b_{m',n'}
	\end{array}$$
    \end{minipage}\quad 
     \begin{minipage}{.48\textwidth}
        $$P \ominus Q = \begin{array}{ccc} a_{1,1} & \cdots & a_{1,n} \\ 
		\vdots & \ddots & \vdots \\ 
		a_{m,1} & \cdots & a_{m,n}\\
		b_{1,1} & \cdots & b_{1,n'} \\ 
		\vdots & \ddots & \vdots \\ 
		b_{m',1} & \cdots & b_{m',n'}
	\end{array}$$
    \end{minipage}
    \caption{Explaining two concatenations of $P = [a_{i,j}]_{m,n}$ and $Q = [b_{i,j}]_{m',n'}$.}
    \label{fig:concats}
\end{figure}
A two-dimensional string or a picture is a two-dimensional rectangular array over $\Sigma$ represented as 
	$ P = \begin{array}{ccc}a_{1,1} & \cdots & a_{1,n} \\ 
		\vdots & \ddots & \vdots \\ 
		a_{m,1} & \cdots & a_{m,n}
	\end{array} = [a_{i,j}]_{m,n}$
	where $m, n\in $\( \mathbb{N} \); $a_{i,j} \in \Sigma$ refers to the $j$th symbol in the $i$th row of $P$ with $1 \leq i \leq m, 1 \leq j \leq n$. The number of rows and columns of the picture is given by $\left|P \right|_r = m, \left|P \right|_c = n$, respectively.	The set of all pictures over $\Sigma$ is denoted by $\Sigma_+^+$. A picture language over $\Sigma$ is a subset of~$\Sigma_+^+$. Let $P = [a_{i,j}]_{m,n}$ and $Q = [b_{i,j}]_{m',n'}$ be two pictures over $\Sigma$. The column concatenation of $P$ and~$Q$ is defined if $m = m'$ and is denoted by $P \obar Q$. The row concatenation of $P$ and~$Q$ is defined if $n = n'$ denoted by $P \ominus Q$. Both are explained in Figure~\ref{fig:concats}. We use power notation like $P^n$ to denote $(n-1)$-fold column catenation of $P$ with itself, and similarly $P_n$ to denote $(n-1)$-fold row catenation of $P$ with itself.
    
In contrast to the 1-dimensional (string or word) case, also geometric operations are important when dealing with arrays. In this paper, we are focusing on the four operations of transpose, written $T$, vertical and horizontal reflections (here, we use subscripts $VR$ and $HR$) and rotation by $180^\circ$, indicated by the superscript~$R$. The operations are explained in Figure~\ref{fig:geometric-ops}.
More operations can be defined according to the dieder group $D_8$, see \cite{FerParTho2018}. Notice that in the case of single-row arrays (that can be viewed as words), $P^R=P_{VR}$ is indeed the reversal operation discussed above. The operations on pictures can naturally be extended to array languages. 

\begin{figure}[tb]
\begin{minipage}{.23\textwidth}
$$T(P) = \begin{array}{ccc}
		a_{1,1} & \cdots & a_{m,1} \\
		\vdots & \ddots & \vdots \\ 
		a_{1,n} & \cdots & a_{m,n}
	\end{array}$$
\end{minipage}\quad \begin{minipage}{.23\textwidth}
$$P_{VR} = \begin{array}{ccc}
		a_{1,n} & \cdots  & a_{1,1} \\ 
		\vdots & \ddots & \vdots \\ 
		a_{m,n} & \cdots & a_{m,1}
	\end{array}$$
\end{minipage}\quad\begin{minipage}{.23\textwidth}
$$P_{HR} = \begin{array}{ccc}
		a_{m,1}  & \cdots & a_{m,n} \\ 
		\vdots  & \ddots & \vdots \\ 
		a_{1,1}  & \cdots & a_{1,n}
	\end{array}$$
\end{minipage}\quad\begin{minipage}{.23\textwidth}
$$P^R = \begin{array}{ccc}
		a_{m,n}  & \cdots & a_{m,1} \\ 
		\vdots  & \ddots & \vdots \\ 
		a_{1,n}  & \cdots & a_{1,1}
	\end{array}$$
\end{minipage}
    \caption{Some geometric unary operations applied to $P = [a_{i,j}]_{m,n}$.}
    \label{fig:geometric-ops}
\end{figure}
 
Next, we introduce two-dimensional grammars according to \cite{SirSirKri72}.

     \begin{definition}
        A two-dimensional right-linear matrix grammar (or RMG) [context-free matrix grammar (or CFMG)] is defined by a 7-tuple $G= (V_h,V_v,\Sigma_I,\Sigma,S,R_h,R_v)$ where $V_h$ is a finite set of horizontal variables,  $V_v$ is a finite set of vertical variables, $\Sigma_I\subseteq V_v$ is a finite set of intermediates,  $\Sigma$ is a finite set of terminals, $S\in V_h$ is the start symbol, $R_h$ is a finite set of horizontal rules of the form $V\rightarrow AV^\prime$ or $V\rightarrow A$ where $V, V^\prime\in V_h$ and $A\in \Sigma_I$ for RMG [or $V\rightarrow B$ where $V\in V_h$ and $B\in (V_h\cup \Sigma_I)^*$ for CFMG], $R_v$ is a finite set of vertical rules of the form $W\rightarrow aW^\prime$ or $W\rightarrow a$ where $W,W^\prime\in V_v$ and $a\in \Sigma$.                       
    \end{definition}
The derivation is carried out in two phases. In the first phase, the string grammar $G_h=(V_h, \Sigma_I,S,R_h)$ generates a string language $H(G)$ over the alphabet $\Sigma_I$. The strings in $H(G)$ form the top row of the picture. In the second phase, treating each intermediate symbol as a start symbol, the vertical generation of the columns of the picture is done in parallel by applying the rules in $R_v$. This parallel application of rules from $R_v$ ensures that rules of the form $V_i\rightarrow a_i$ are all applied simultaneously in each column. These grammars make sure that the columns can grow only in one direction.

    \begin{definition}
    A Returning Finite Automaton (RFA) is defined as a 7-tuple $M=(Q,\Sigma,\delta,s,F,\#,\square)$ where $Q$ is the finite set of states, $\Sigma$ is the input alphabet, $\delta:Q\times(\Sigma\cup \{\#\})\to 2^Q$ is the transition function, $s\in Q$ is the start state, $F$ is the set of final states, $\#\notin \Sigma$ serves as the boundary symbol and $\square$ indicates an erased (i.e., an already visited) position.
\end{definition}
An RFA processes the input picture row by row always from left to right as shown in \figurename~\ref{fig:RFA}. It should be noted that after the automaton reads the rightmost $\#$ in a row, it continues to read the first element after the leftmost $\#$ in the next row. A picture is accepted by a RFA if the automaton begins in the start state, executes the computation according to the transition function and enters a final state after reading the last symbol of the picture (which is not $\#$). The roles of $\#$ and $\square$ will be similar in other automaton models presented in this paper.

\begin{example}\label{ex:LsubL}
  The set $L_{\textrm{L}}$ of tokens $\textrm{L}$ of different sizes and proportions formally defined as $$L_{\textrm{L}} =\left\{(X\obar(\bullet)^{n})_{m}\ominus X^{n+1}\mid  n,m\geq 1 \right\}\,.$$   can be accepted by the RFA $M_{\textrm{L}}=(\{s,s_1,s_2,s_3\},\{X,\bullet\},\delta ,\{s\},\{s_2\},\#,\square)$, with $\delta$ defined by  $(s,X)\mapsto \{s_3\}$, $(s_3,\bullet)\mapsto \{s_1\}$, $(s_1,\bullet)\mapsto \{s_1\}$, $(s_1,\#)\mapsto\{s,s_2\}$, $(s_2,X)\mapsto \{s_2\}$. Notice that $M_{\textrm{L}}$ does not accept any single-row array.
\end{example} 

\begin{example}\label{ex:Lww} 
The language $L_{w\obar w}=\{(a[1]\obar a[1])\ominus(a[2]\obar a[2])\ominus\cdots\ominus(a[m]\obar a[m])\mid m\geq 1,\,\forall i\in\{1,\dots,m\}: a[i]\in\{0,1\}\}$ can be accepted by an RFA with states $s,s_0,s_1,s_\#$ and transition function $(s,0)\mapsto\{s_0\}$, $(s,1)\mapsto\{s_1\}$, $(s_0,0)\mapsto\{s_\#\}$, $(s_1,1)\mapsto\{s_\#\}$, $(s_\#,\#)\mapsto\{s\}$ with start and final state~$s$.
\end{example}

\begin{figure}
    \centering
    \scalebox{.95}{\begin{minipage}{.35\textwidth}\mbox{ }\\[-5ex]
    
$\begin{array}{ccccccc}
		\#&\underset{\rightarrow}{0}& \underset{\rightarrow}{0}& \underset{\rightarrow}{0} &
        \underset{\rightarrow}{0}& \underset{\rightarrow}{0}&
        \underset{\hookleftarrow}{\underset{\rightarrow}{\#}}\\ 
        
 	\#&\underset{\hookrightarrow}{0}& \underset{\rightarrow}{0}& 
    \underset{\rightarrow}{0} &
        \underset{\rightarrow}{0}& \underset{\rightarrow}{0}&
       \underset{\hookleftarrow}{\underset{\rightarrow}{\#}}\\ 
        
	\#&\underset{\hookrightarrow}{1}& \underset{\rightarrow}{1}& \underset{\rightarrow}{1} &
        \underset{\rightarrow}{0}& \underset{\rightarrow}{1}& \underset{\hookleftarrow}{\underset{\rightarrow}{\#}}\\ 
        
    {\#}&\underset{\hookrightarrow}{0}& \underset{\rightarrow}{0}& \underset{\rightarrow}{0} &
        \underset{\rightarrow}{0}& \underset{\rightarrow}{0}&\underset{\hookleftarrow}{\underset{\rightarrow}{\#}}\\  
        
      {\#}&\underset{\hookrightarrow}{0}& \underset{\rightarrow}{0}& \underset{\rightarrow}{0} &
        \underset{\rightarrow}{0}& \underset{\rightarrow}{0}&\#\\

\end{array}	$	
 \end{minipage}}\quad\scalebox{.95}{\begin{minipage}{.35\textwidth}
$\begin{array}{ccccccc}
		\#&\square& \square& \square &
       \square& \square&
        \# \\ 
        
 	\#&\square& \square& 
    \square &\square& \square&\# \\ 
        
	\#&\square& \square& \square &
       \underset{\rightarrow} {0}&\underset{\rightarrow}{1}&\underset{\hookleftarrow}{\underset{\rightarrow}{\#}}\\ 
        
     \#&\underset{\hookrightarrow}{0}& \underset{\rightarrow}{0}&\underset{\rightarrow}{0} &
        \underset{\rightarrow}{0}& \underset{\rightarrow}{0}&\underset{\hookleftarrow}{\underset{\rightarrow}{\#}}\\ 
        
     \#&\underset{\hookrightarrow}{0}&\underset{\rightarrow}{0}&\underset{\rightarrow}{0} &\underset{\rightarrow}{0}&\underset{\rightarrow}{0}&
     \#\\

\end{array}	$
 \end{minipage}}\\[1ex]    
    \caption{RFA processing an input}
    \label{fig:RFA}
\end{figure} 
\begin{definition}
     A Returning Pushdown Automaton (RPDA) is defined as a 7-tuple $M=(Q,\Sigma, \Gamma, \delta,q_{0},\linebreak[3] \startofpushdown,\#,\square)$ where
  $Q$ is the finite set of states, $\Sigma$ is the input alphabet with $\#\notin\Sigma$, 
  $\Gamma$ is the pushdown alphabet, $\delta: Q\times(\Sigma \cup \{\#,\varepsilon\})\times \Gamma \rightarrow 2^{Q \times \Gamma^{*}}$ is the transition function, $q_{0}\in Q$ is the initial state and $\startofpushdown \in\Gamma$ is the start symbol of the pushdown store.  		
\end{definition}

 An RPDA also processes the input picture row by row, always from left to right. A computation by an RPDA is accepted if it has an empty pushdown store after completely processing the picture.
 
 If $X$ is any model (automaton or grammar) for describing pictures, then $\mathcal{L}(X)$ is the family of all non-empty picture languages that can be described by $X$.
 
 We see that $\mathcal{L}(\mathrm{RFA})= T(\mathcal{L}(\mathrm{RMG}))$ in \cite{FerPST2018}. Later, Fernau \emph{et al.} \cite{FerJen2025} studied the relations among the three classes  $\mathcal{L}(\mathrm{RPDA})$, $T(\mathcal{L}(\mathrm{CFMG}))$ and $\mathcal{L}(\mathrm{CFMG})$ since RPDA and CFMG are the extended models of RFA and RMG, respectively.
 We recall the following results from \cite{FerJen2025}.
 \begin{theorem}
      $ T(\mathcal{L}(\mathrm{CFMG})) \subsetneq \mathcal{L}(\mathrm{RPDA})$. 
 \end{theorem}
    \begin{proposition}\label{prp:CFMG-APDA}
The two array language families $ \mathcal{L}(\mathrm{RPDA})$  and 
		$\mathcal{L}(\mathrm{CFMG}) $ are incomparable.
\end{proposition} 

\section{Two-Head Returning Finite Automata}
In this section, we introduce our new model: the
Returning Finite Automata with 2 heads. 
\begin{definition}
   A 2-Head Returning Finite Automaton, or 2-HRFA for short, is a 7-tuple $M=(Q,\Sigma,\delta,s,\linebreak[3]F,\#,\square)$ where $Q$ is the finite set of states, $\Sigma$ is the input alphabet, $s\in Q$ is the start state, $F$ is the set of final states and $\delta: Q\times ((\Sigma\cup \{\varepsilon,\#\})\times (\Sigma\cup \{\varepsilon,\#\})\setminus\{(\varepsilon,\varepsilon)\})\rightarrow 2^Q$ is the transition function. The symbol $\#\notin \Sigma$ denotes the boundary of the picture and $\square$ denotes a position passed by the automaton. 
\end{definition}

\noindent\textbf{Working of $M$:}\\
The automaton processes an input picture $P\in \Sigma_m^n$ working on $P\in (\Sigma\cup \{\square\})^+_+$ surrounded by the boundary symbol $\#$ such that $P^\#\in\{\#_m\}\obar \Sigma_{m}^{n}\obar \{\#_m\}$. Each of the two heads of the automaton reads the input picture row by row, respectively, moving in opposite directions. That is, the first head begins to read the picture from the top moving from left to right and the second head reads the picture from the bottom moving from right to left on each row. A picture is accepted by a 2-HRFA when the entire picture is processed  (i.e., the two heads meet) and the 2-HRFA enters a final state. This is formally detailed in the following.
\medskip

\noindent\textbf{Configuration of  $M$:}\\
We represent a configuration by $(q,P_\square,\mu,\mu^\prime)$, where $q\in Q$ is the current state, $P_\square\in (\Sigma\cup \{\square\})^+_+$ is the remaining picture still to be worked on, $\mu$ and $\mu^\prime$ are integers such that the first head is in the $\mu$th row and the second head is in the $\mu^\prime$th row. Hence, if $P_\square$ has $m$ rows, then $1\leq \mu\leq \mu'\leq m$.   We denote the currently read symbols by the first head and the second head as $(a_{\mu,j},b_{\mu^\prime,j^\prime})$,  respectively, and their positions in $P_\square^\#$ as $(\mu,j)$ and $(\mu^\prime,j^\prime)$,  respectively, where $a_{\mu,j},b_{\mu^\prime,j^\prime}\in \Sigma\cup \{\#,\varepsilon\},0\leq j,j^\prime\leq n+1$ 
(this indexing of $P^\#_\square$ allows to have a more natural indexing of $P_\square$ itself) and $1\leq \mu\leq\mu^\prime\leq m$. 
Actually, instead of the letter being in the given position, it may be allowed (depending on the transition function)  reading~$\varepsilon$, meaning that the given head(s) stay(s) in the same position in this computation step and the symbol under the head is actually ignored and not yet read (by the given head). Hence, $(q,P_\square,\mu,\mu^\prime)$, with $q\in Q$, is \emph{valid} if $1\leq \mu\leq \mu'\leq m$ and the following conditions are satisfied.
\begin{itemize}
    \item The  first $\mu-1$ rows and the last $m-\mu^{\prime}$ rows contain only $\square$.
    \item If $\mu<\mu^{\prime\prime}<\mu^{\prime}$, then the $\mu^\text{th}$ row of $P_\square^\#$ is in $\{\#\}\obar \{\square\}^j\obar \Sigma^{n-j}\obar \{\#\}$, $\mu^{\prime\prime\,\text{th}}$ row is contained in $\{\#\}\obar \Sigma^n\obar \{\#\}$ and the $\mu^{\prime}{\,}^\text{th}$ row  of $P_\square^\#$ is in $\{\#\}\obar \Sigma^{n-j^\prime}\obar \{\square\}^{j^\prime}\obar \{\#\}$ for some $0\leq j,j^\prime\leq n$. Now, the current position (in $P_\square^\#$) of the first head is $(\mu,j+1)$ and that of the second head is $(\mu^\prime,n-j^\prime)$. 
\item If $\mu=\mu^{\prime}$, then the  $\mu^\text{th}$ row of $P_\square^\#$ is in $\{\#\}\obar \{\square\}^j\obar \Sigma^{n-(j+j')}\obar \{\square\}^{j'}\obar \{\#\}$ for some $0\leq j,j'\leq n$ and $j+j'\leq n$. 
\end{itemize}
When the automaton starts the computation, the first head is in the first row reading the first symbol of $P$ and the second head is placed in the  $m^\text{th}$ row reading the last symbol of $P$. Accordingly, when we try to process the picture~$P$, 
$(s,P,1,m)$ is the initial configuration. The set of final configurations is $\{(q,\square_m^n,\mu,\mu) \mid q\in F\land m,n\geq 1, 1\leq\mu\leq m\}$.
\medskip

\noindent
\textbf{Valid configurations and  configuration transitions of $M$:} 
\begin{itemize}
    \item  Two-head moves
    
    Let $(q,P_\square,\mu,\mu^{\prime})$ and $(p,P_\square^\prime,\mu,\mu^{\prime})$ be two valid configurations such that $P_\square$ and $P_\square^\prime$ are the same except for the symbols $a_{\mu,j},b_{\mu^
\prime,j^\prime}$ in positions $(\mu,j)$ and $(\mu^{\prime},j^\prime)$, respectively, and these symbols from $\Sigma$ in $P_\square$ are replaced by $\square$ in $P_\square^\prime$. Then, $(q,P_\square,\mu,\mu^{\prime})\vdash_M (p,P_\square^\prime,\mu,\mu^{\prime})$ if $p\in \delta(q,a_{\mu,j},b_{\mu^
\prime,j^\prime})$.

In these steps it is required that the two heads are not in the same position, i.e., they can read the letter in their position without affecting the other heads position.
\item One-head moves
\begin{itemize}
    \item Let $(q,P_\square,\mu,\mu^{\prime})$ and $(p,P_\square^\prime,\mu,\mu^{\prime})$ be two valid configurations such that $P_\square$ and $P_\square^\prime$ are the same except for a symbol $b_{\mu^
\prime,j^\prime}$ in  position $(\mu^{\prime},j^\prime)$; this symbol from $\Sigma$ in $P_\square$ is replaced by~$\square$ in $P_\square^\prime$. Then,
    $(q,P_\square,\mu,\mu^{\prime})\vdash_M (p,P_\square^\prime,\mu,\mu^{\prime})$ if $p\in \delta(q,\varepsilon,b_{\mu^
\prime,j^\prime})$.
\item 
      Let $(q,P_\square,\mu,\mu^{\prime})$ and $(p,P_\square^\prime,\mu,\mu^{\prime})$ be two valid configurations such that $P_\square$ and $P_\square^\prime$ are the same except for a symbol $a_{\mu,j}$ in position $(\mu,j)$; this symbol from $\Sigma$ in $P_\square$ is replaced by~$\square$ in $P_\square^\prime$. Then,
    $(q,P_\square,\mu,\mu^{\prime})\vdash_M (p,P_\square^\prime,\mu,\mu^{\prime})$ if $p\in \delta(q,a_{\mu,j},\varepsilon)$.
\end{itemize}
        
\item One-head end of row with two-head moves 
\begin{itemize}
    \item Let $(q,P_\square,\mu,\mu^{\prime})$ and $(p,P_\square^\prime,\mu+1,\mu^{\prime})$ be two valid configurations such that $P_\square$ and $P_\square^\prime$ are the same except for a symbol $b_{\mu^
\prime,j^\prime}$ in  position $(\mu^{\prime},j^\prime)$; this symbol from $\Sigma$ in $P_\square$ is replaced\linebreak[4] by~$\square$ in $P_\square^\prime$. Moreover, $a_{\mu,j}=\#$ is scanned by the first head. Then,
    $(q,P_\square,\mu,\mu^{\prime})\vdash_M (p,P_\square^\prime,\mu+1,\mu^{\prime})$ if $p\in \delta(q,\#,b_{\mu^
\prime,j^\prime})$.
\item Let $(q,P_\square,\mu,\mu^{\prime})$ and $(p,P_\square^\prime,\mu,\mu^{\prime}-1)$ be two valid configurations such that $P_\square$ and $P_\square^\prime$ are the same except for a symbol $a_{\mu,j}$ in position $(\mu,j)$; this symbol from $\Sigma$ in $P_\square$ is replaced by $\square$ in $P_\square^\prime$. Moreover, $b_{\mu',j'}=\#$ is scanned by the second head.  Then,
    $(q,P_\square,\mu,\mu^{\prime})\vdash_M (p,P_\square^\prime,\mu,\mu^{\prime}-1)$ if $p\in \delta(q,a_{\mu,j},\#)$.
    \end{itemize}
    In the following cases, when the heads read $\#$ symbol, then in fact only its row is changing in the configuration, the picture with letters and squares does not change.
\item End of row and one-head moves 
\begin{itemize}
     \item Let $(q,P_\square,\mu,\mu^{\prime})$ and $(p,P_\square,\mu+1,\mu^{\prime})$ be valid configurations  with $\mu<\mu'$ and $p\in \delta(q,\#,\varepsilon)$, then    $(q,P_\square,\mu,\mu^{\prime})\vdash_M (p,P_\square,\mu+1,\mu^{\prime})$ if $a_{\mu,j}=\#$ is scanned by the first head.
\item Let $(q,P_\square,\mu,\mu^{\prime})$ and $(p,P_\square,\mu,\mu^{\prime}-1)$ be valid configurations   with $\mu<\mu'$ and $p\in \delta(q,\varepsilon,\#)$, then    $(q,P_\square,\mu,\mu^{\prime})\vdash_M (p,P_\square,\mu,\mu^{\prime}-1)$ if $b_{\mu',j'}=\#$ is scanned by the second head.
\end{itemize}
  
\item Both heads at the end of row 
\begin{itemize}
    \item 
Let $(q,P_\square,\mu,\mu^{\prime})$ and $(p,P_\square,\mu+1,\mu^{\prime}-1)$ be valid configurations with $\mu<\mu'-1$ and $p\in \delta(q,\#,\#)$. Then
    $(q,P_\square,\mu,\mu^{\prime})\vdash_M (p,P_\square,\mu+1,\mu^{\prime}-1)$ if both heads scan $\#$.
\end{itemize}
\end{itemize}
\textbf{Language acceptance by $M$:}\\
Thus, the language accepted by $M$ is 
$$L(M)=\{P\in \Sigma_+^+\mid (s,P,1,m)\vdash^*_M(f,\square_m^n,\mu,\mu)\mid f\in F; n,m\geq 1,1\leq \mu\leq m\}\,.$$

The automaton can be represented graphically by a transition diagram. The arc between any two states has a label which is a pair of symbols $(a,b)$. This means that the symbol $a$ is read by the first head and the symbol $b$ is read by the second head where any of $a$ or $b$ could be $\varepsilon$.

The 2-HRFA $M$ is deterministic, or 2-HRDFA for short, if $|\delta(q,a,b)|\leq 1$ for each $q\in Q$ and $(a,b)\in (\Sigma\cup \{\#\})\times (\Sigma\cup \{\#\})$. Moreover, if there are one-head moves of the first head, e.g.,  $\delta(q,a,\varepsilon)=p$, then there is no letter $b\in\Sigma\cup\{\#\}$ with $\delta(q,a,b)=r$ or $\delta(q,\varepsilon,b)=r$ for any $r\in Q$. 
Further, in case there are one-head moves by the second head, e.g.,
$\delta(q,\varepsilon,a)=p$, then there is no letter $b\in\Sigma\cup\{\#\}$ with $\delta(q,b,a)=r$ or $\delta(q,b,\varepsilon)=r$ for any $r\in Q$.

\begin{figure}
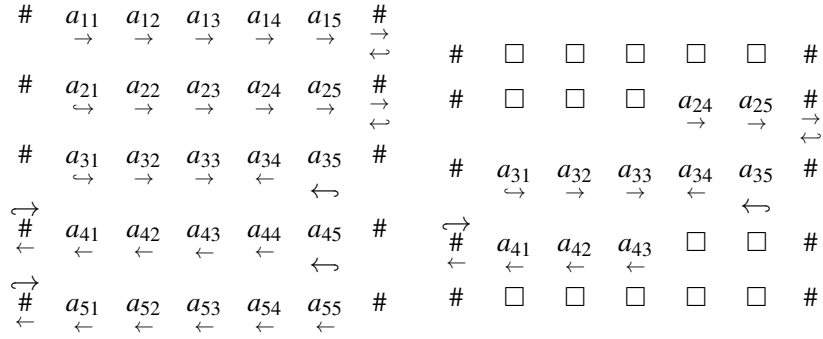

    \centering
    \scalebox{.95}{\begin{minipage}{.35\textwidth}\mbox{ }\\[-5ex]
    
$\begin{array}{ccccccc}
		\#&\underset{\rightarrow}{a_{11}}& \underset{\rightarrow}{a_{12}}& \underset{\rightarrow}{a_{13} } &
        \underset{\rightarrow}{a_{14}}& \underset{\rightarrow}{a_{15}}&
        \underset{\hookleftarrow}{\underset{\rightarrow}{\#}}\\
        
 	\#&\underset{\hookrightarrow}{a_{21}}& \underset{\rightarrow}{a_{22}}& 
    \underset{\rightarrow}{a_{23}} &
        \underset{\rightarrow}{a_{24}}& \underset{\rightarrow}{a_{25}}&
        \underset{\hookleftarrow}{\underset{\rightarrow}{\#}}\\
        
	\#&\underset{\hookrightarrow}{a_{31}}& \underset{\rightarrow}{a_{32}}& \underset{\rightarrow}{a_{33}} &
        \underset{\leftarrow}{a_{34}}& \underset{\scalebox{1}[-1]{$\hookleftarrow$}}{a_{35}}&
        {\#} \\ 
        
     \overset{\scalebox{1}[-1]{$\hookrightarrow$}}{\underset{\leftarrow} {\#}}&\underset{\leftarrow}{a_{41}}& \underset{\leftarrow}{a_{42}}& \underset{\leftarrow}{a_{43}} &
        \underset{\leftarrow}{a_{44}}& \underset{\scalebox{1}[-1]{$\hookleftarrow$}}{a_{45}}&{\#} \\ 
        
     \overset{\scalebox{1}[-1]{$\hookrightarrow$}}{\underset{\leftarrow} {\#}}&\underset{\leftarrow}{a_{51}}& \underset{\leftarrow}{a_{52}}& \underset{\leftarrow}{a_{53}} &
        \underset{\leftarrow}{a_{54}}&
  
\underset{\leftarrow}{a_{55}}&{\#} \\                      
\end{array}	$	
 \end{minipage}}\quad\scalebox{.95}{\begin{minipage}{.35\textwidth}
$\begin{array}{ccccccc}
		\#&\square& \square& \square &
       \square& \square&
        \# \\ 
        
 	\#&\square& \square& 
    \square &\underset{\rightarrow}{a_{24}}& \underset{\rightarrow}{a_{25}}& \underset{\hookleftarrow}{\underset{\rightarrow}{\#}} \\ 
        
	\#&\underset{\hookrightarrow}{a_{31}}&\underset{\rightarrow}{a_{32}}& \underset{\rightarrow}{a_{33}} &\underset{\leftarrow}{a_{34}}&\underset{\scalebox{1}[-1]{$\hookleftarrow$}}{a_{35}}&\# \\ 
        
    \overset{\scalebox{1}[-1]{$\hookrightarrow$}}{\underset{\leftarrow} {\#}}&\underset{\leftarrow}{a_{41}}&\underset{\leftarrow}{a_{42}}&\underset{\leftarrow}{a_{43}} &
        \square& \square&\# \\ 
        
     \#&\square&\square&\square &\square& \square&\# \\ 
                      
\end{array}	$
 \end{minipage}}\\[1ex]    
    \caption{2-HRFA processing an input of size $5\times 5$.}
    \label{fig:2-HRFA}
\end{figure}

Any 2-HFA can be interpreted as a 2-HRFA. Upon doing so, this 2-HRFA will accept only single-row arrays because there is no transition on the border symbol~$\#$ provided. Notice that in particular no 2-row array can be accepted, because we require that in a final configuration, both heads are at the same row, so that they can meet. This is only possible if one of the heads has digested a border symbol.
  
    \begin{example}\label{ex:Lsubrev}
     The language $L_{rev}=\left\{w \ominus w^R\mid w\in\{a,b\}^+ \right\}$ is accepted by the deterministic 2-HRFA~$M_{\textrm{rev}}$ (given in \figurename~\ref{fig:Lrev}).
            Notice the similarity to \autoref{ex:Lsub{w=wtor}}.
\end{example}

\begin{figure}[t]
		\centering
		% States
				\begin{tikzpicture}[shorten >=1pt, node distance=7em, on grid, auto,thick,initial text=]

% Define states
					\node[state, initial] (s) {$s$};
			          \node[state, right=of s,accepting] (s1) 
                    {$s_1$};
                    
% Transitions
					\draw[->] (s) edge [loop above] node[align=center]{$(a,a)$\\$(b,b)$} ();
                    \draw[->] (s) edge node[align=center]{$(\#,\varepsilon)$} (s1);
                    
                \end{tikzpicture}
		\caption{2-HRFA $M_{rev}$ that recognizes the language  $L_{rev}$. \label{fig:Lrev}}
	\end{figure} 
 \begin{example}\label{ex:LsubH}
    The set $L_{\textrm{H}}$ of tokens $\textrm{H}$ of different sizes and same proportions (maintaining horizontal symmetry of the figure) formally defined as $$L_{\textrm{H}} =\left\{(X\obar(\bullet)^{n}\obar X)_{m}\ominus X^{n+2}\ominus(X\obar(\bullet)^{n}\obar X)_{m}\mid  n,m\geq 1 \right\}\,,$$ 
    $$i.e., \scalebox{1}{$L_{\textrm{H}} =\left\{\begin{array}{ccc}
	 X& \bullet& X \\ 
      X& X& X  \\
	 X& \bullet& X
\end{array},\quad\begin{array}{cccc}
	 X& \bullet&\bullet& X \\ 
      X& X& X &X \\
	 X& \bullet&\bullet& X
\end{array},\quad \begin{array}{ccccc}
	 X& \bullet&\bullet&\bullet& X \\ 
      X& X& X &X&X \\
	 X& \bullet&\bullet&\bullet& X
\end{array},\quad\begin{array}{ccccc}
	 X& \bullet&\bullet& \bullet&X \\
     X& \bullet&\bullet& \bullet& X \\
      X& X& X &X& X \\
      X& \bullet&\bullet& \bullet& X \\
	 X& \bullet&\bullet& \bullet& X
\end{array},...\right\}$}$$
    can be accepted by a 2-HRFA, $M_{\textrm{H}}$ (see \figurename~\ref{fig:LsubH}).
\end{example}

    \begin{figure}[t]
		\centering
		% States
				\begin{tikzpicture}[shorten >=1pt, node distance=7em, on grid, auto,thick,initial text=]

% Define states
					\node[state, initial] (s) {$s$};
					\node[state, right=of s] (q1) {$s_1$};
					\node[state, right=of s1] (s2) {$s_2$};
                    \node[state, right=of s2] (s3) {$s_3$};
                     \node[state, right=of s3,accepting] (s4) {$s_4$};
                  ;
% Transitions
					\draw[->] (s) edge node[align=center] {$(X,X)$} (s1);
                    \draw[->] (s1) edge node[align=center] {$(\bullet,\bullet)$} (s2);
                     \draw[->] (s2) edge [loop above]node[align=center] {$(\bullet,\bullet)$} ();                  
                    \draw[->] (s2) edge node[align=center] {$(X,X)$} (s3);
                    \draw[->] (s3) edge node[align=center] {$(\#,\#)$} (s4);
                     \draw[->] (s3) edge[bend left] node[align=center] {$(\#,\#)$} (s);
                    \draw[->] (s4) edge [loop above]node[align=center] {$(X,X)$\\$(X,\varepsilon)$} ();
                    \end{tikzpicture}
		\caption{ A 2-head returning finite automaton $M_{\textrm{H}}$ that accepts the language $L_{\textrm{H}}$. \label{fig:LsubH}}
	\end{figure}
  
\begin{example}\label{ex:Lsub010}
         The language $L_{(010)}$ of rectangles of odd side lengths with $1$ in the center and $0$ elsewhere. 
 is accepted by the 2-HRFA $M_{010}=(Q, \Sigma, \delta, s, F,\#,\square)$ (see \figurename~\ref{fig:Lsub010}) where $Q =  \{s,s_{1},s_{2},s_{3},s_{4},s_{5}\}$, $\Sigma = \{0,1\}$, $s\in Q$ is the start state,  $F=\{s_{5}\}$ is the final state.

  \begin{example}\label{ex:a3b1}
    Suppose the automaton $M_{a^{3k}b^k}$ has 4 states: $\{q,p,r,s\}$ with initial and only accepting state~$s$ described by the transition diagram in \figurename~\ref{fig:Lsuba3b1}. Some of the members of the language $L_{a^{3k}b^k}$ accepted by $M_{a^{3k}b^k}$ are: $aaab$, $aaaaaabb$, as well as 
$$\begin{array}{cc}
	a& a \\ 
     a& b \\
	 \end{array}, \ \ \begin{array}{cccc}
	  a & a & a & a\\
    a & a & b & b
\end{array}, \ \ \begin{array}{cccc}
	  a & a & a & a\\
    a & a & a & a\\
    a & b & b & b
\end{array}, \ \ \begin{array}{cccccccc}
	  a & a & a & a &a & a & a & a\\
    a & a & a & a &a & a & a & a\\
    a & a & b & b & b & b& b & b
\end{array}, \ \ \begin{array}{c}
	  a \\
    a \\
    a \\
    b
\end{array}, \ \ 
\begin{array}{cc}
	  a & a \\
    a & a\\
    a & a\\
    b & b
\end{array}.$$
    \end{example}

     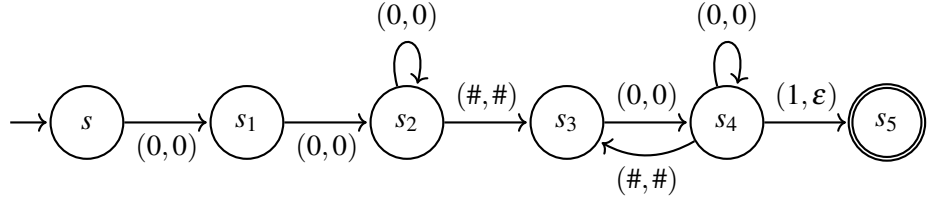
\begin{figure}[t]
		\centering
		 
	\begin{tikzpicture}[shorten >=1pt, node distance=5.5em, on grid, auto,thick,initial text=]

% Define states
					\node[state, initial] (s) {$s$};
					\node[state, right=of s] (s1) {$s_1$};
                    \node[state, right=of s1] (s2) {$s_2$};
                    \node[state, right=of s2] (s3) {$s_3$};
                     \node[state, right=of s3] (s4) {$s_4$};
                      \node[state, right=of s4,accepting] (s5) {$s_5$};
% Transitions
					\draw[->] (s) edge node[align=center,below] {$(0,0)$} (s1);
                    \draw[->] (s1) edge node[align=center,below] {$(0,0)$} (s2);
                    \draw[->] (s2) edge[loop above] node[align=center] {$(0,0)$} ();
                    \draw[->] (s2) edge node[align=center] {$(\#,\#)$} (s3);
                    \draw[->] (s3) edge node[align=center] {$(0,0)$} (s4);
                    \draw[->] (s4) edge[loop above] node[align=center] {$(0,0)$} ();
                    \draw[->] (s4) edge [bend left]node[align=center] {$(\#,\#)$} (s3);
                    \draw[->] (s4) edge node[align=center] {$(1,\varepsilon)$} (s5);

                \end{tikzpicture}
		\caption{2-HRFA $M_{010}$ that recognizes the language  $L_{010}$. \label{fig:Lsub010}}
	\end{figure}

 \begin{figure}[t]
		\centering		 
	\begin{tikzpicture}[shorten >=1pt, node distance=7em, on grid, auto,thick,initial text=]

% Define states
					\node[state, initial,accepting] (s) {$s$};
					\node[state, right=of s] (p) {$p$};
                    \node[state, right=of p] (q) {$q$};
                    \node[state, right=of q] (r) {$r$};
% Transitions
					\draw[->] (s) edge node[align=center] {$(a,\varepsilon)$} (p);
                    \draw[->] (s) edge[loop above] node[align=center] {$(\#,\varepsilon)$} ();
                    \draw[->] (p) edge node[align=center] {$(a,\varepsilon)$} (q);
                    \draw[->] (p) edge[loop above] node[align=center] {$(\#,\varepsilon)$} ();
                    \draw[->] (q) edge node[align=center] {$(a,\varepsilon)$} (r);
                    \draw[->] (q) edge [loop above]node[align=center] {$(\#,\varepsilon)$} ();
                    \draw[->] (r) edge[bend left] node[align=center,above] {$(\varepsilon,b)$} (s);
                    \draw[->] (r)edge[loop above] node[align=center] {$(\varepsilon,\#)$} ();
                    
                \end{tikzpicture}
		\caption{2-HRFA $M_{a^{3k}b^k}$ that recognizes the language $L_{a^{3k}b^k}$\label{fig:Lsuba3b1}.}
	\end{figure}
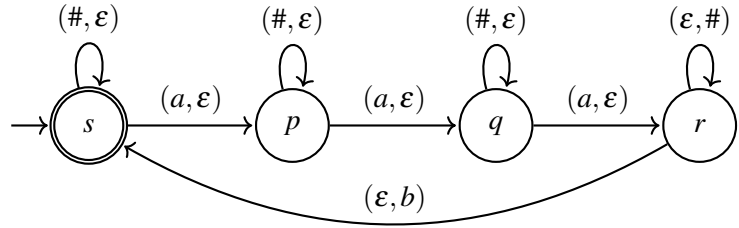     

In the next example, we show that the new model is able to accept special languages that are usually hard for most of the models like RFA  and CFMG due to pumping lemmas (Iteration Theorems) in \cite{FerPST2018,KriSir74a} that we now quote for the convenience of the reader. 
\begin{theorem}
    Let $L$ be a context-free matrix language (CFML). Then there exist integers $r$ and $t$ such that, for every $P\in L$ with $\left|P \right|_c = n, n>r$, there is a decomposition $P=U\obar V\obar X\obar Y\obar Z$ where $\left|V\obar Y \right|_c \geq 1$ such that $U\obar V^k\obar X\obar Y^k\obar Z$ is in $L$ for all $k\geq 1$.
\end{theorem}

\begin{example}\label{ex:square}
         The language $L_{\text{squares}}$ of unary squares  
 is accepted by the 2-HRFA $M_{\text{squares}}=(\{s,q ,f\}, \{a\},\linebreak[3] \delta, s, \{f\},\#,\square)$  where $q= \delta(s,a,\varepsilon)$, $q=\delta(q,\varepsilon,a)$, $s=\delta(q,\varepsilon,\#)$, $f=\delta(s,\#,\varepsilon)$.
 For each step of the first head, the second head reads a full row, in this way, while the first head counts the columns of the input, the second head counts the rows. The input is accepted if these numbers match, as they are counted in parallel. Notice, that $M_{\text{squares}}$ is, in fact,  deterministic.
\end{example}
    \end{example}

\noindent
With  \autoref{ex:Lsubrev}, more precisely its transpose, we immediately see:
    
       \begin{theorem}\label{thm:RFA-vs-2HRFA}      $\mathcal{L}\mathrm{(RFA)}=T(\mathcal{L}\mathrm{(RMG)})\subsetneq \mathcal{L}\mathrm{(2\mbox{-}HRFA)}$.
    \end{theorem}

    \begin{theorem}
        \label{thm:2-hrfavsRPDA}
       $\mathcal{L}\mathrm{(2\mbox{-}HRFA)}\subsetneq \mathcal{L}\mathrm{(RPDA)}$.
    \end{theorem} 
       \begin{proof}
        First let us consider the one-row picture language $L_1= \{a^j \obar b^j\obar a^k\obar b^k\mid j,k\geq 1\}$. Clearly, $L_1$ can be accepted by an RPDA by pushing $j$ $a$'s in to the storage and while reading the $j$ $b$'s, it pops the $j$ $a$'s from the storage and repeats the same for the $k$ $a$'s and $k$ $b$'s. But $L_1$ is not linear and cannot be accepted by a 2-HFA. Thus, $L_1\notin \mathcal{L}\mathrm{(2\mbox{-}HRFA)}$ which proves the strictness of inclusion. The inclusion itself can be seen as in the one-dimensional case, compare Theorem~\ref{thm:2HFA}.
    \end{proof}
    
 \begin{definition}
        Let $M=(Q, \Sigma, \delta, s, F,\#,\square)$ be a 2-HRFA and $L(M)\subseteq \Sigma_+^+$. Let $\Sigma_{\#}=\Sigma\cup\{\#\}$.     
        For an array language $L\subseteq \Sigma_+^+$, define $L_{\text{flat}}=\left\{w(1)\#w(2)\#\cdots \#w(m)\in \Sigma_\#^*\mid w(1)\ominus w(2)\ominus \cdots \ominus w(m)
%        \begin{bmatrix}w_1\\w_2\\\vdots\\ w_m\end{bmatrix}
    \in L\right\}$ as the \emph{flattened version} of~$L$ (consisting of \emph{standard strings} in the terminology of~\cite{GiaRes97}). 
    \end{definition}
    \begin{theorem}\label{thm:2HRFA-vs-RMG}
     $\mathcal{L}\mathrm{(2\mbox{-}HRFA)}$ and  $\mathcal{L}\mathrm{(RMG)}$ are incomparable.
    \end{theorem}
\begin{proof}
       The language $L_{col}=\left\{w\ominus w\mid w\in \{a,b\}^+\right\}\in \mathcal{L}(\mathrm{RMG})$ (by \cite{FerJen2025}, Example 3) but $L_{col}\notin \mathcal{L}\mathrm{(RPDA)}$ (\cite{FerJen2025},Theorem 3) because the associated flattened language $\{w\#w\mid w\in \{a,b\}^+\}$ cannot be accepted by a pushdown automaton equipped with a blind counter that makes one turn. Hence, by \autoref{thm:2-hrfavsRPDA} $L_{col}\notin \mathcal{L}\mathrm{(2\mbox{-}HRFA)}$.  
            Next, consider the language $L_{rev}$. From \autoref{ex:Lsubrev}, we see that $L_{rev}\in \mathcal{L}\mathrm{(2\mbox{-}HRFA)}$ but $L_{rev}\notin \mathcal{L}\mathrm{(RMG)}$ due to column iteration theorem given in \cite{KriSir74a}.
        \end{proof}
       \begin{theorem}\label{thm:2HRFA-vs-CFMG}
     $\mathcal{L}\mathrm{(2\mbox{-}HRFA)}$ and  $\mathcal{L}\mathrm{(CFMG)}$ are incomparable.
    \end{theorem}
    \begin{proof}
        By a similar argument as in \autoref{thm:2HRFA-vs-RMG}, we see that the language $L_{col}\in \mathcal{L}(\mathrm{CFMG})$ but cannot be recognized by a 2-HRFA. On the other hand, the language $L_{010}$ in \autoref{ex:Lsub010} can be accepted by a 2-HRFA but $L_{010} \notin \mathcal{L}\mathrm{(CFMG)}$ due to the iteration theorems given in  \cite{KriSir74a}. 
    \end{proof}

    \begin{theorem}\label{thm:2hrfa-T(cfmg)}
        $\mathcal{L}\mathrm{(2\mbox{-}HRFA)}$ and  $T(\mathcal{L}\mathrm{CFMG})$ are incomparable. 
           \end{theorem} 
    
     \begin{proof}
        Consider the language $L_{(010)}=T(L_{(010)})$. By \autoref{ex:Lsub010}, $L_{(010)}\in \mathcal{L}(\mathrm{(2\mbox{-}HRFA)})$, but $L_{(010)}\notin \mathcal{L}(\mathrm{CFMG})$ due to the iteration theorems given in~\cite{KriSir74a}.
        Now, consider the single row picture language $ L= \{a^j\obar b^j\obar a^k\obar b^k\mid j,k\geq 1\}$ which is a context-free language and hence a context-free matrix language. %
        $T(L)=\left\{
        %\begin{bmatrix}a_j\\b_j\\a_k\\b_k\end{bmatrix}
        a_j\ominus b_j\ominus a_k\ominus b_k
        \mid j,k\geq 1\right\}\in T(\mathcal{L}\mathrm{(CFMG)})$ but $T(L)\notin \mathcal{L}\mathrm{(2\mbox{-}HRFA)}$.  Hence, the two classes of language are incomparable.
    \end{proof}

    \section{Both-Head Stepping 2-Head Returning Finite Automata}
    In this section, we define a variant of the 2-HRFA called both-head stepping 2-Head RFA or in short, B2-HRFA, similar to the definition of 2-HRFA by a 7-tuple $M=(Q,\Sigma,\delta,s,F,\#,\square)$.
The B2-HRFA processes an input picture surrounded by the boundary symbol $\#$ analogous to the 2-HRFA except that at each step of the computation, both the heads must move one step. In case the input picture is left with only one unread symbol, then only the first head moves and reads the input and accepts the picture if it enters the final state. 

 A transition with the ``symbol pair'' $(a,\varepsilon)_1$ denotes that this transition is allowed only when the heads may read the same position (filled with $a$) to finish the input. For instance, for a picture with an even number of rows, $(\#,\varepsilon)_1$ is required. Otherwise, this exceptional type of transition is necessary for an odd number of columns.

 The B2-HRFA is deterministic, or B2-HRDFA for short, if for each $q^\prime\in Q$ and $(a,b)\in (\Sigma\cup \{\#\})\times (\Sigma\cup \{\varepsilon,\#\})$, there is at most one $p^\prime\in Q$ such that $\delta(q^\prime,(a,b))=p^\prime$. (The $\varepsilon$-move of the second head is reserved to the last step.)
 The language $L_{rev}$ and $L_\textrm{L}$ can be accepted by a B2-HRFA. (See   \autoref{fig:Lrev}.)
In the one-dimensional case, B2-HRFA are known to describe the even-linear languages, see \cite{AmaPut64,FerSem00,NagyiConcept2010,SemGar94}.

\begin{theorem}\label{thm:rfa-b2hrfa-2hrfa}
        $\mathcal{L}\mathrm{(RFA)}\subsetneq \mathcal{L}\mathrm{(B2\mbox{-}HRFA)}\subsetneq \mathcal{L}\mathrm{(2\mbox{-}HRFA)}$
    \end{theorem}
    \begin{proof}
    The language $L_{rev}\in  \mathcal{L}\mathrm{(B2\mbox{-}HRFA)}$ (see \figurename~\ref{fig:Lrev}) but cannot be accepted by a RFA due to the pumping lemma given in \cite{FerPST2018}. Also, by \autoref{ex:a3b1}, the language $L_{a^{3k}b^k}\in \mathcal{L}\mathrm{(2\mbox{-}HRFA)}$ but not in $\mathcal{L}\mathrm{(B2\mbox{-}HRFA)}$. The reason is as follows: Consider the one-dimensional case of the language $L_{a^{3k}b^k}$ which is $\{a^{3k}b^k\mid k\geq 1\}$. Clearly this set is not even-linear and hence cannot be accepted by a both head stepping 2-HFA \cite{NagyiConcept2010} even though it is accepted by a 2-HFA. 

    To prove the inclusion $\mathcal{L}\mathrm{(RFA)}\subseteq \mathcal{L}\mathrm{(B2\mbox{-}HRFA)}$, let us consider the following construction. Let $L\in \mathcal{L}\mathrm{(RFA)}$ and $M=(Q,\Sigma,\delta,s,F,\#,\square)$ be the RFA accepting $L$ where $\delta\subseteq Q\times (\Sigma\cup \{\#\})\times Q$. Due to nondeterminism, we can assume that $F=\{s_f\}$ is a  singleton set.
We construct a B2-HRFA $\hat{M}=(\hat{Q},\Sigma, \hat{\delta},\hat{s},\hat{F},\#,\square)$ to accept $L$, with $\hat{Q}=(Q\times Q)\cup\{f\}$, $\hat{s}=(s,s_f)$. %$\hat{s}=(s,F)$.
For $q,p\in Q$ and $a,b\in \Sigma\cup \{\#\}$,  %$P\subseteq Q$,  
let $\hat\delta((q,p),(a,b))=\{(q',p')\mid q'\in\delta(q,a),p\in\delta(p',b)\}\}$.
Moreover, in order to digest an odd number of columns, we add $\hat\delta((q,p),(a,\varepsilon)_1)=\{f\mid p\in \delta(q,a)\}$, and to digest an even number of rows, we add $\hat\delta((q,p),(\#,\varepsilon)_1)=\{f\mid p\in \delta(q,\#)\}$.
Let $\hat F=\{f\}\cup \{(p,p)\mid p\in Q\}$. An induction proof to show the correctness of the construction is now straightforward.
    \end{proof}
    \begin{illustration}
        We illustrate the construction of B2-HRFA in \autoref{thm:rfa-b2hrfa-2hrfa} with the RFA defined in \autoref{ex:LsubL} by a transition diagram with only useful states in \figurename~\ref{fig:LsubL}.
        \end{illustration}
       
    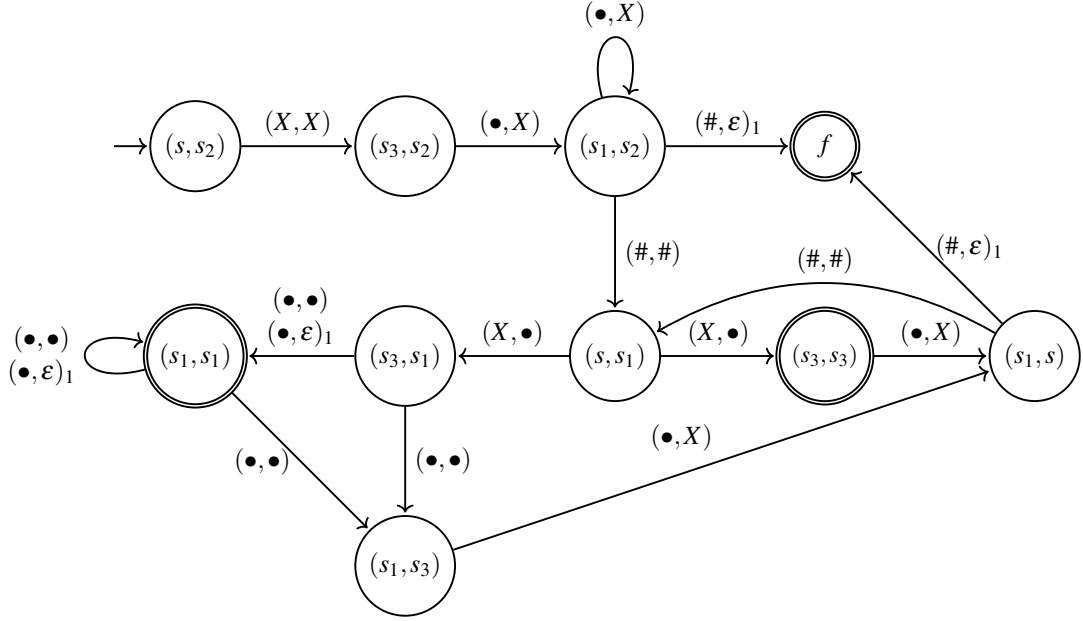
\begin{figure}[t]
		\centering
		  \resizebox{0.9\linewidth}{!}{%
\begin{tikzpicture}[shorten >=1pt, node distance=8em, on grid, auto,thick,initial text=]

% Define states
					\node[state, initial] (ss2) {$(s,s_2)$};
					\node[state, right=of ss2] (s3s2) {$(s_3,s_2)$};
                    \node[state, right=of s3s2] (s1s2) {$(s_1,s_2)$};
                    \node[state, right=of s1s2,accepting] (f) {$f$};
                    \node[state, below=of s1s2] (ss1) {$(s,s_1)$};
                     \node[state, left=of ss1] (s3s1) {$(s_3,s_1)$};
                      \node[state, left=of s3s1, accepting](s1s1) {$(s_1,s_1)$};
                       \node[state, below=of s3s1] (s1s3) {$(s_1,s_3)$};                      
                       %\node[state,  left=of s1s,accepting] (ss) {$(s,s)$};
                     \node[state,  right=of ss1,accepting] (s3s3) {\small $(s_3,s_3)$};
                      \node[state,  right=of s3s3] (s1s) {$(s_1,s)$};
                    
% Transitions
					\draw[->] (ss2) edge node[align=center,above] {$(X,X)$} (s3s2);
                    \draw[->] (s3s2) edge node[align=center,above] {$(\bullet,X)$} (s1s2);
                    \draw[->] (s1s2) edge[loop above] node[align=center] {$(\bullet,X)$} ();
                    \draw[->] (s1s2) edge node[align=center,above] {$(\#,\varepsilon)_1$} (f);
                     \draw[->] (s1s2) edge node[align=center,right] {$(\#,\#)$} (ss1);
                    \draw[->] (ss1) edge node[align=center,above] {$(X,\bullet)$} (s3s1);
                    \draw[->] (ss1) edge node[align=center,above] {$(X,\bullet)$} (s3s3);
                    \draw[->] (s3s1) edge node[align=center,above] {$(\bullet,\bullet)$\\$(\bullet,\varepsilon)_1$} (s1s1);
                    \draw[->] (s1s1) edge[loop left] node[align=center] {$(\bullet,\bullet)$\\$(\bullet,\varepsilon)_1$} ();
                    \draw[->] (s3s1) edge node[align=center] {$(\bullet,\bullet)$} (s1s3);
                    \draw[->] (s1s3) edge node[align=center] {$(\bullet,X)$} (s1s);
                    \draw[->] (s3s3) edge node[align=center,above] {$(\bullet,X)$} (s1s);
                    \draw[->] (s1s) edge [bend right]node[align=center,above] {$(\#,\#)$} (ss1);
                    \draw[->] (s1s) edge node[align=center,right] {$(\#,\varepsilon)_1$} (f);
                     \draw[->] (s1s1) edge 
                    node[align=center,left] {$(\bullet,\bullet)$} (s1s3);

                \end{tikzpicture}
                }
		\caption{B2-HRFA $M^\prime_{\textrm{L}}$ constructed by \autoref{thm:rfa-b2hrfa-2hrfa} recognizing the language  $L_{\textrm{L}}$ (from \autoref{ex:LsubL}) with only useful states; we sent the transitions with $(\bullet,\varepsilon)_1$ not to $f$ but to $(s_1,s_1)$ to simplify the drawing.} \label{fig:LsubL}
	\end{figure}

     \begin{theorem}
          $\mathcal{L}\mathrm{(B2\mbox{-}HRFA)}=\mathcal{L}\mathrm{(B2\mbox{-}HRDFA)}$ 
     \end{theorem}
     \begin{proof}
         The proof is based on the subset construction similar to the models RFA for arrays \cite{FerPST2018} and both-head stepping two-head automata defined for strings \cite{NagyiConcept2010}.
     \end{proof}
      \begin{theorem}\label{thm:B2HRFA-vs-CFMG}
     $\mathcal{L}\mathrm{(B2\mbox{-}HRFA)}$ and  $\mathcal{L}\mathrm{(CFMG)}$ are incomparable.
    \end{theorem}
    \begin{proof}
        By a similar argument as in \autoref{thm:2HRFA-vs-CFMG}, we see that the two classes are incomparable.
    \end{proof}

\begin{example}\label{exa:0..010..0}
It is not hard to see that the array language $L=\{0\}^+_+\ominus\{1\}^+\ominus\{0\}^+_+$ can be described by an RFA.
We claim that the language $L'=\{0\}^+_+\obar\{1\}_+\obar\{0\}^+_+$ cannot be described by any B2-HRFA. Suppose the contrary. Then, some B2-HRFA  $M'=(Q,\{0,1\},\delta,s,F,\#,\square)$ accepts $L'$ and hence accepted by $M'$. Let $n=|Q|$ and consider processing $A(i)=(0^i\obar 1\obar 0^{n-i})_{2n+2}$ for $i=0,\dots,n$. Notice that all $A(i)$ are in $L'$. By the pigeon-hole principle, there is a state $p(i)$ that is entered twice after reading $\#$ with both heads, not entering the last row to be read. Again by the same principle, among these $n+1$ many states $p(0),\dots,p(n)$, there must be a state $q\in Q$ listed twice, say, $q=p(j)=p(j')$, with $j<j'$. This means that $M'$ can process an array as follows: It first works as if processing $A(j)$, but after first seeing $p(j)$  after reading $\#$ with both heads, it switches towards working as if it was processing $A(j')$ after first seeing $p(j')$ after reading $\#$ with both heads. This way, we see that $M'$ must also accept an array that is clearly not in~$L'$.
\end{example}

      \section{Closure and Decidability Properties}

      By a standard argument, one sees:
    \begin{theorem}\label{thm:Closure-union}
        $\mathcal{L}\mathrm{(2\mbox{-}HRFA)}$ is closed under union.
    \end{theorem}

    \begin{theorem}\label{thm:Closure-intersection}
        $\mathcal{L}\mathrm{(2\mbox{-}HRFA)}$ is not closed under intersection.
    \end{theorem}
    \begin{proof}
      From the string case it is clear that the languages $L_a=\{a^nb^nc^i\mid n\geq 1,i\geq 1\}$ and $L_b= \{a^ib^nc^n\mid n\geq 1,i\geq 1\}$ are linear context-free languages. From \autoref{thm:2HFA}, we see that $L_a$ and $L_b$ can be accepted by 2-head finite automata (and hence by 2-HRFA). But their intersection language $L= L_a\cap L_b=\{a^nb^nc^n\mid n\geq 1,i\geq 1\}$ is not a context-free language~\cite{HopUll79}. Thus by \autoref{thm:2HFA}, $L$ cannot be accepted by a 2-HFA. Also, $L$ cannot be accepted by a 2-HRFA. Hence $\mathcal{L}\mathrm{(2\mbox{-}HRFA)}$ is not closed under intersection. 
    \end{proof}
    \noindent
    By De Morgan's law, the previous two results immediately imply:
     \begin{theorem}\label{thm:Closure-complement}
        $\mathcal{L}\mathrm{(2\mbox{-}HRFA)}$ is not closed under complement.
    \end{theorem}
    
\begin{theorem}\label{thm:Closure-concat}
        $\mathcal{L}\mathrm{(2\mbox{-}HRFA)}$ is not closed under row and column concatenation.
    \end{theorem}
    \begin{proof}
    The language $L_1= L\obar L$ where $L=\{a^jb^j\mid j\geq 1\}$ cannot be accepted by a 2-HRFA whereas the language $L$ is linear and can be accepted by a 2-HRFA. Similar argument follows for $L_2= L\ominus L$.
    \end{proof}
    \begin{theorem}\label{thm:b2hrfa-unionintersectioncomp}
        $\mathcal{L}\mathrm{(B2\mbox{-}HRFA)}$ is closed under union,  intersection and complement. 
    \end{theorem}
    \begin{proof}
    Closure under union is again easy. Since deterministic and non-deterministic models of B2-HRFA define the same class of languages, we see that $\mathcal{L}\mathrm{(B2\mbox{-}HRFA)}$ is closed under complement. By De Morgan's law, closure under intersection follows.        
    \end{proof}
     
     \begin{theorem}\label{thm:b2hrfa-VR}
        $\mathcal{L}\mathrm{(B2\mbox{-}HRFA)}$ is closed under vertical reflection. 
    \end{theorem}
    \begin{proof}        
     Let $L\in \mathcal{L}\mathrm{(B2\mbox{-}HRFA)}$ such that $L=L(M)$ for a B2-HRFA $M$. We construct another B2-HRFA, $M^\prime$, to accept the vertical reflection language $L_{VR}$ of $L$. Suppose $M$ accepts a picture starting in the state $s$ reading the rows $\mu,\mu^\prime$ by the first and the second head, respectively, and reaches a state say $q$ on reaching the end of the rows. Now, $M^\prime$ stores the state $s$ in its finite control and guesses the state $q$. Then, $M^\prime$ starts to scan the rows $\mu,\mu^\prime$ from state $q$ by applying the transitions of $M$ in the reverse direction and when $M^\prime$ reaches the end of the rows, the automaton proceeds to compute only if $M^\prime$ is in state $s$ and this procedure is repeated on all the remaining rows. Formally, the states of $M^\prime$ are like $(r,p,t)$ where $p$ is the actual state, $t$ is the state to be reached by $M^\prime$ and $r$ is the state at which $M^\prime$ starts simulating $M$ in each pair of rows processed by the two heads simultaneously.
     
          There is an exception of this processing if an input array with an odd number of rows is processed. The middle row is where the two heads meet, and then we switch back to the state set~$Q$ of~$M$ and simply reverse the role of the two heads.
   Thus, let $M=(Q,\Sigma,\delta,s,F,\#,\square)$ be some B2-HRFA. We construct $M^\prime=(Q^\prime,\Sigma,\delta^\prime,s^\prime,F^\prime,\#,\square)$ where $Q^\prime=(Q\times Q\times Q)\cup Q\cup \{s'\}$, $\delta^\prime(s',a,b)=\{(r,p,s)\mid r\in Q, p\in \delta(r,a,b) \}\cup \{p\in Q\mid p\in\delta(s,b,a)\}$ for $a,b\in \Sigma$, $\delta^\prime((r,p,t),a,b)=\{ (r,q,t)\mid p\in \delta(q,a,b)\}$ for $r,p,t\in Q$ and
    $a,b\in \Sigma$, $\delta^\prime(q,a,b)=\{p\in Q\mid p\in\delta(q,b,a)\}$ for $q\in Q$, $a,b\in \Sigma$, $\delta^\prime (q,(a,\varepsilon)_1)=\delta(q,(a,\varepsilon)_1)$ for $q\in Q$ and $a\in\Sigma$, 
    $\delta^\prime((r,t,t),\#,\#)=\{(p,p,r)\mid p\in Q\}$ for  $r,t\in Q$. With this setting, we make sure that after reading the current upper and lower rows completely, the guessed state is verified. Therefore, $F_{\text{even}}'=\{(r,t,t)\mid t\in Q, r\in F\}$ is part of the final states of $M'$ (for an even number of rows). For an odd number of rows, we add $F_{\text{odd}}'=F$, so that $F'=F_{\text{even}}'\cup F_{\text{odd}}'$. 
    
    Let us add a few words on the correctness proof of this construction (that can be given by induction based on these arguments and the description of the work of $M'$ given above). The role of the start symbol~$s'$ can be best understood by noticing that it is assumed to cover two cases: (a) If the input array has a single row, then it will enter a state~$p$ based on swapping the roles of the two heads (compared to~$M$). This is the same situation as entered when $M'$ guesses it enters the last row to be processed from both ends after digesting the border symbol(s). (b) If the input array has more than one row, it will first guess the state that $M$ would enter after processing the first and last rows and then also perform reading the first pair of symbols. Again, this is very similar to what happens after digesting the border symbol(s) and then start processing the next pair of rows later on. Having understood this role of $s'$, one sees that this was introduced only due to the technical fact that the model that we have introduced does not allow for sets of start symbols. 
    \end{proof}        

\begin{example}
    The language $L_{\textrm{E}}$ of tokens $\textrm{E}$ with odd number of rows and even number of columns defined as $$L_{\textrm{E}} =\left\{(X)_{m}\obar(X\ominus (\bullet)_{\frac{m-3}{2}}\ominus X\ominus (\bullet)_{\frac{m-3}{2}}\ominus X)^{n}\mid n=2t-1, m=2k+3,k,t\geq 1 \right\}\,.$$ is accepted by a B2-HRFA $M_\textrm{E}=(\{s,s_1,s_2,s_3,s_4\},\{X,\bullet\},\delta,\{s\},F=\{s_4\})$, with transitions defined by $\delta(s,X,X)=\{s\}$, $\delta(s,\#,\#)=\{s_1\}$, $\delta(s_1,X,\bullet)=\{s_2\}$, $\delta(s_2,\bullet,\bullet)=\{s_2\}$, $\delta(s_2,\bullet,X)=\{s_3\}$, $\delta(s_3,\#,\#)=\{s_1,s_4\}$, $\delta(s_4,X,X)=\{s_4\}$. An illustration of the  construction of B2-HRFA $M^\prime_\textrm{E}$ in \autoref{thm:b2hrfa-VR} with the B2-HRFA $M_\textrm{E}$ is given by a transition diagram with only useful states in \figurename~\ref{fig:LsubEsubVR}. Notice that we allowed ourselves to omit states that are not (co-)reachable. Also, we allowed ourselves to merge the two states $(s_3,s_3,s)$ with $(s_3,s_3,s_1)$.
\end{example}
       
\begin{figure}[t]
		\centering
		% States
         \resizebox{0.9\textwidth}{!}{%
				\begin{tikzpicture}[shorten >=1pt, node distance=7.6em, on grid, auto,thick,initial text=]

% Define states
					\node[state, initial] (sa) {$s^\prime$};
					\node[state, right=of sa] (s) {$(s,s,s)$};
					\node[state, right=of s] (s3) {$(s_3,s_3,s)$};
                    \node[state, right=of s3] (s2) {$(s_3,s_2,s)$};
                    \node[state, right=of s2] (s1) {$(s_3,s_1,s_1)$};
                     \node[state, right=of s1,accepting] (s4) {$s_4$};
                  ;
% Transitions
					\draw[->] (sa) edge node[align=center] {$(X,X)$} (s);
                    \draw[->] (s) edge node[align=center] {$(\#,\#)$} (s3);
                     \draw[->] (s) edge [loop above]node[align=center] {$(X,X)$} (); 
                    \draw[->] (s3) edge node[align=center] {$(\bullet,X)$} (s2);
                    \draw[->] (s2) edge [loop above]node[align=center] {$(\bullet,\bullet)$} ();
                     \draw[->] (s2) edge node[align=center] {$(X,\bullet)$} (s1);
                      \draw[->] (s1) edge[bend left] node[align=center] {$(\#,\#)$} (s3);
                    \draw[->] (s1) edge node[align=center] {$(\#,\#)$} (s4);
                     \draw[->] (s4) edge [loop above]node[align=center] {$(X,X)$} ();
                    \end{tikzpicture}
                    }
		\caption{B2-HRFA $M^\prime_\textrm{E}$ that accepts  $L_{{\textrm{E}_{VR}}}$.} \label{fig:LsubEsubVR}
	\end{figure}
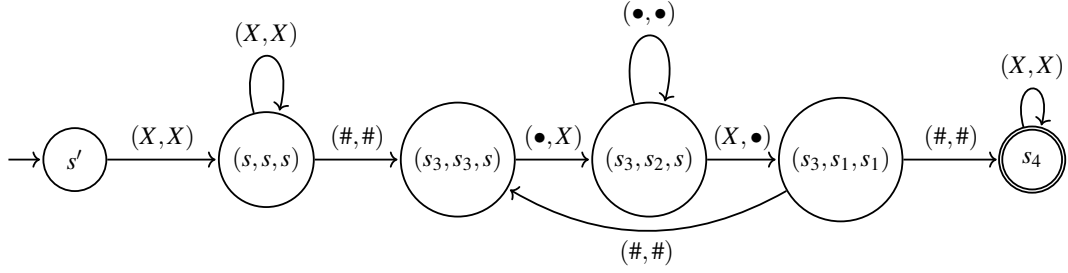  
     \begin{theorem}\label{thm:b2hrfa-rotation}
       Both  $\mathcal{L}\mathrm{(2\mbox{-}HRFA)}$ and $\mathcal{L}\mathrm{(B2\mbox{-}HRFA)}$ are closed under rotation by $180^\circ$.
     \end{theorem}
           
    \begin{proof}
         Closure under rotation by $180^\circ$ can be achieved by interchanging each arc label $(a,b)$ of symbols between any two states to $(b,a)$ in the given 2-HRFA and B2-HRFA ($a,b\in \Sigma \cup \{\varepsilon, \#\}$) except the possible last transitions of B2-HRFA where only the first head is used. 
    \end{proof} 
    
 \begin{theorem} \label{thm:b2hrfa-HR}
          $\mathcal{L}\mathrm{(B2\mbox{-}HRFA)}$ is closed under horizontal reflection.
     \end{theorem}
     \begin{proof}
         Horizontal reflection is nothing but vertical reflection followed by rotation by $180^\circ$. Therefore, closure under horizontal reflection follows from \autoref{thm:b2hrfa-VR} and \autoref{thm:b2hrfa-rotation}.
     \end{proof}

     \begin{theorem}\label{thm:b2hrfa-rotation90}
       Neither $\mathcal{L}\mathrm{(B2\mbox{-}HRFA)}$ nor $\mathcal{L}\mathrm{(2\mbox{-}HRFA)}$ is closed under rotation by $\pm 90^\circ$ nor under transpose.
     \end{theorem}

\begin{proof}
    By \autoref{ex:Lww} and \autoref{thm:rfa-b2hrfa-2hrfa}, there is a language in $\mathcal{L}\mathrm{(B2\mbox{-}HRFA)}$ whose rotation by $\pm 90^\circ$ is not in  $\mathcal{L}\mathrm{(2\mbox{-}HRFA)}$ by the proof of \autoref{thm:2HRFA-vs-RMG}. The same argument applies to transpose. A self-contained example for $\mathcal{L}\mathrm{(B2\mbox{-}HRFA)}$ is presented in \autoref{exa:0..010..0}.
\end{proof}

Based on these arguments, (non-)closure results for other geometric unary operations can be derived using the well-known properties of the dihedral group $D_8$ and its subgroups, see \cite{FerParTho2018} and any textbook on group theory.  Let us at least briefly look into decidability problems connected to B2-HRFA.

\begin{remark}
By the proof of \autoref{thm:rfa-b2hrfa-2hrfa}, decidability problems cannot be simpler than those for RFA as studied in \cite{FerPST2018}. However, there is also a simple way how to re-interpret arrays over the alphabet $\Sigma$ as arrays over $\Sigma\times\Sigma\cup \Sigma$.
     For a given input array $A$ of size $m\times (n+2)$, say $\begin{array}{ccccc}\#&a_{1,1} & \cdots & a_{1,n} &\#\\ 
		\vdots&\vdots & \ddots & \vdots&\vdots \\ 
		\#&a_{m,1} & \cdots & a_{m,n}&\#
	\end{array}$, we construct a new array $B$ by compressing the original array $A$ in the following manner whose size is $\dfrac{m}{2}\times (n+1)$ if $m$ is even and $\dfrac{m+1}{2}\times (n+1)$ if $m$ is odd. That is,

    $$B=\begin{array}{ccccc}
    \left(\begin{smallmatrix}
         a_{1,1} \\
         a_{m,n}          
    \end{smallmatrix}\right)
			 &   \left(\begin{smallmatrix}
         a_{1,2} \\
         a_{m,n-1}          
    \end{smallmatrix}\right)
   & \cdots&  \left(\begin{smallmatrix}
         a_{1,n} \\
         a_{m,1}          
    \end{smallmatrix}\right) 
    &  \left(\begin{smallmatrix}
         \# \\
         \#          
    \end{smallmatrix}\right) \\
     \left(\begin{smallmatrix}
         a_{2,1} \\
         a_{m-1,n}          
    \end{smallmatrix}\right)
			 &  \left(\begin{smallmatrix}
         a_{2,2} \\
         a_{m-1,n-1}          
    \end{smallmatrix}\right)
   & \cdots&  \left(\begin{smallmatrix}
         a_{2,n} \\
         a_{m-1,1}          
    \end{smallmatrix}\right) 
    &\left(\begin{smallmatrix}
         \# \\
         \#          
    \end{smallmatrix}\right)\\
            \vdots & \vdots & \ddots & \vdots \\ 
			 \left(\begin{smallmatrix}
         a_{\frac{m}{2},1} \\
         a_{\frac{m}{2}+1,n}          
    \end{smallmatrix}\right)
             &  \left(\begin{smallmatrix}
         a_{\frac{m}{2},2} \\
         a_{\frac{m}{2}+1,n-1}         
    \end{smallmatrix}\right)
             & \cdots&  \left(\begin{smallmatrix}
         a_{\frac{m}{2},n} \\
         a_{\frac{m}{2}+1,1}          
    \end{smallmatrix}\right)
    &\left(\begin{smallmatrix}
         \# \\
         \epsilon         
    \end{smallmatrix}\right)
		\end{array}\quad \text{if $m$ is even and}$$ 
        $$B=\begin{array}{ccccc}
    \left(\begin{smallmatrix}
         a_{1,1} \\
         a_{m,n}          
    \end{smallmatrix}\right)
			 &   \left(\begin{smallmatrix}
         a_{1,2} \\
         a_{m,n-1}          
    \end{smallmatrix}\right)
   & \cdots&  \left(\begin{smallmatrix}
         a_{1,n} \\
         a_{m,1}          
    \end{smallmatrix}\right) 
    &  \left(\begin{smallmatrix}
         \# \\
         \#          
    \end{smallmatrix}\right) \\
     \left(\begin{smallmatrix}
         a_{2,1} \\
         a_{m-1,n}          
    \end{smallmatrix}\right)
			 &  \left(\begin{smallmatrix}
         a_{2,2} \\
         a_{m-1,n-1}          
    \end{smallmatrix}\right)
   & \cdots&  \left(\begin{smallmatrix}
         a_{2,n} \\
         a_{m-1,1}          
    \end{smallmatrix}\right) 
    &\left(\begin{smallmatrix}
         \# \\
         \#          
    \end{smallmatrix}\right)\\
            \vdots & \vdots & \ddots & \vdots \\ 
    a_{\frac{m+1}{2},1}&  a_{\frac{m+1}{2},2}& \cdots&  a_{\frac{m+1}{2},n}&\#
		\end{array}\quad \text{if $m$ is odd.}$$
        This transformation allows us to interpret the work of a B2-HRFA as the work of a RFA (over a different alphabet), so that (e.g.) the non-emptiness problem is (also) \textsf{NP}-complete, etc.
The corresponding decision problems for 2-HRFA are widely open. Let us finally remark that we know (unpublished) that the non-emptiness problem is undecidable for RPDA.  Also, it is known that this problem is decidable for CFMG.
\end{remark}

 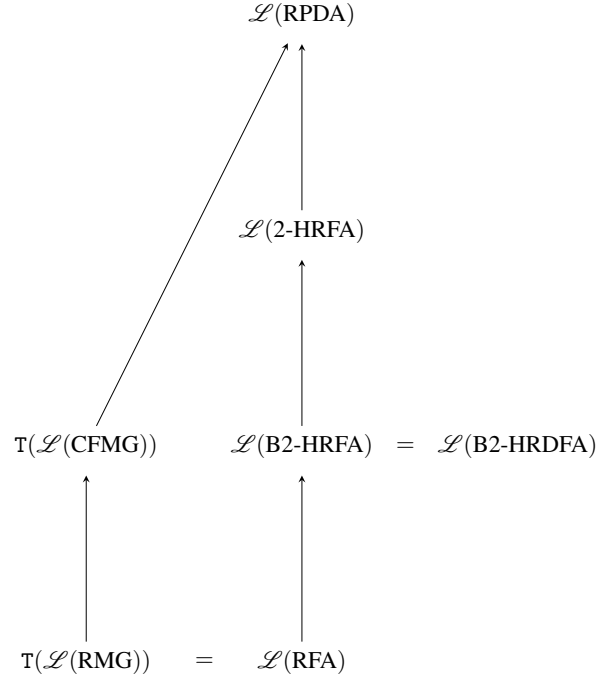
\begin{figure}[t]
    \begin{center}
    \resizebox{0.5\linewidth}{!}{%
       \begin{tikzpicture}[>=latex, node distance=9em, scale=0.85, every node/.style={}]

		% Nodes
        \node (RFA)   {$\mathcal{L} (\mathrm{RFA})$}; 
        \node (TRMG)     [left of=RFA] {$\mathtt{T}(\mathcal{L} (\mathrm{RMG}))$}; 
        % \node (RFA0)    [right of=RFA] {}; 
        \node (B2HRFA)[above of=RFA]{$\mathcal{L} (\mathrm{B2\mbox{-}HRFA})$};
        \node (TCFMG) [above of =TRMG]{$\mathtt{T}(\mathcal{L} (\mathrm{CFMG}))$};
       \node (RFA1)[right of=B2HRFA]{$\mathcal{L} (\mathrm{B2\mbox{-}HRDFA})$};
        %\node (CFMG0)[above of=TCFMG]{};
        \node (2HRFA)[above of=B2HRFA]{$\mathcal{L} (\mathrm{2\mbox{-}HRFA})$};
        \node (RPDA)[above of=2HRFA]{$\mathcal{L} (\mathrm{RPDA})$};

		  \draw[-stealth,shorten >=4pt] (RFA) -- (B2HRFA);
        \draw[-stealth,shorten >=4pt] (TRMG) -- (TCFMG);
        \draw[-stealth,shorten >=4pt] (TCFMG) -- (RPDA);
        \draw[-stealth,shorten >=4pt] (B2HRFA) -- (2HRFA);
        \draw[-stealth,shorten >=4pt] (2HRFA) -- (RPDA);       
			
		% Equal signs
		 \path (TRMG.east) -- node[ inner sep=2pt] {$=$} (RFA.west);

\path (B2HRFA.east) -- node[fill=white, inner sep=2pt] {$=$} (RFA1.west);
		\end{tikzpicture}
        }
    \end{center}		
    \caption{Hierarchy among classes of array languages recognized by two-dimensional automaton models mentioned in this paper. An arrow $\mathcal{L}(X)\rightarrow \mathcal{L}(Y)$ indicates the relation $\mathcal{L}(X)\subsetneq \mathcal{L}(Y)$.  Families that are not connected by directed path are incomparable.  \label{fig:Hierarchy}}
\end{figure}

   \section{Conclusion} 
     We have explored the role of two heads and constrained movement patterns in returning finite‑state models for two‑dimensional picture languages. Our study of 2‑HRFA and its both‑head stepping variant B2-HRFA demonstrates that interaction between head directions and synchronization can yield strictly stronger language classes than by RFA, while still remaining below the generative power of returning pushdown automata (see \autoref{fig:Hierarchy}). The incomparability with both regular and context‑free matrix grammars suggests that head‑based scanning strategies capture a different type of structural dependency from the derivation mechanisms employed by matrix grammars. As directions for future work, we plan to investigate decidability questions further, and extend the study to two-head boustrophedon finite automata (2-HBFA) whose very definition already poses some problems.  

%\bibliographystyle{eptcs}
%\bibliography{generic}

\end{document}